\documentclass[11pt,onecolumn,twoside]{IEEEtran} 
\usepackage{amssymb,url,amsmath,amsthm}
\usepackage{here}
\usepackage{cases}
\usepackage[pdftex]{graphicx}
\usepackage{bm}
\usepackage{caption}
\theoremstyle{theorem}
\newtheorem{theorem}{Theorem}

\newtheorem{corollary}[theorem]{Corollary}
\newtheorem{lemma}[theorem]{Lemma}

\theoremstyle{definition}
\newtheorem{definition}{Definition}
\newtheorem{example}{Example}
\newtheorem{remark}{Remark}

\def\dmu{\mathrm{d}\mu}
\def\P{\mathcal{P}}
\def\F{\mathcal{F}}
\def\D{\mathcal{D}}
\newcommand{\eqdef}{:=}

\newcommand{\argmin}{\mathop{\rm arg~min}\limits}

\title{Convex Optimization on Functionals of Probability Densities}
\author{Tomohiro Nishiyama \\ Email: htam0ybboh@gmail.com}

\begin{document} 
\maketitle
\bibliographystyle{plain}
\begin{abstract}
In information theory, some optimization problems result in convex optimization problems on strictly convex functionals of probability densities.
In this note, we study these problems and show conditions of minimizers and the uniqueness of the minimizer if there exist a minimizer.

\end{abstract}
\noindent \textbf{Keywords:} convex optimization, entropy, divergence, Lagrangian

\section{Introduction}
In information theory, major quantities such as the Shannon entropy \cite{cover2012elements}, the relative entropy (Kullback-Leibler divergence) \cite{kullback1951information} and $f$-divergence \cite{csiszar1967information} are strictly convex or concave functionals of probability densities \cite{nishiyama2019monotonically}. Optimizing these quantities under some constraints are important problems in various fields including information theory, machine learning, physics, and finance. Some optimization problems result in convex optimization problems on strictly convex functionals of probability densities. For example, the negative Shannon entropy is a strictly convex functional and minimization of the negative Shannon entropy under linear constraints gives one of the results of the maximum entropy method \cite{conrad2004probability}.                                                                   
 
Csisz{\'a}r and Mat{\'u}\v{s} \cite{csiszar2008minimization} studied minimization problems of strictly convex integral functionals of probability densities under linear equality constraints. In our previous note \cite{nishiyama2020minimization}, we studied minimization problems of strictly convex functionals of probability densities under a specified inequality constraint (divergence balls) and a linear equality constraints.   
 
In this note, we generalize these results and discuss convex optimization problems \cite{boyd2004convex} of strictly convex functionals of probability densities.
We show conditions of minimizers and the uniqueness of the minimizer if there exist a minimizer. Furthermore, we show an application example of the result and some examples of the strictly convex functionals.

\section{Preliminaries}
This section provides definitions and notations which are used in this note.
Let $(\Omega, \F)$ be a measurable space where $\Omega$ denotes the sample space and $\F$ denotes the $\sigma$-algebra of measurable events.
Let $\mu$ be a dominating measure of probability measures $P$ (i.e., $P\ll \mu$), and let $p\eqdef \frac{dP}{d\mu}$ denote the $\mu$-densities of $P$. 
Let $\P$ be the set of probability densities, and let $p=q$ denote $p=q \mbox{ a.s. }$.

\begin{definition}[Strictly convex functional]
\label{def_strictly_functional}
Let $p,q\in\P$ and $p\neq q$.
The functional $F[p]: \P\rightarrow \mathbb{R}$ is strictly convex if
\begin{align}
(1-t)F[p] +t F[q] > F[(1-t)p+t q].
\end{align} 
for all $t\in(0,1)$
\end{definition}

\begin{definition}[Convex functional]
\label{def_convex_functional}
Let $p,q\in\P$.
The functional $G[p]: \P\rightarrow \mathbb{R}$ is convex if
\begin{align}
\label{def_eq_strictly_functional1}
(1-t)G[p] +t G[q] \geq G[(1-t)p+t q].
\end{align} 
for all $t\in[0,1]$
\end{definition}

\begin{definition}[Differentiable functional]
Let $p\in\P$ and let $F[p]: \P\rightarrow \mathbb{R}$ be a functional.
The functional $F[p]$ is differentiable if the functional derivative \cite{engel2013density} exists with respect to $p$.
The functional derivative of $F[p]$ with respect to $p$, denoted $\frac{\delta F[p]}{\delta p(z)}(p(z), z)$, is defined as
\begin{align}
\int \frac{\delta F[p]}{\delta p(z)}(p(z), z) \eta(z) \dmu(z) \eqdef \left.\frac{d}{d\epsilon}D[p+\epsilon \eta]\right|_{\epsilon=0},
\end{align}
where $\eta$ is an arbitrary function and the integral is defined on $\Omega$.
\end{definition}
We define $\frac{\delta F[p]}{\delta p(z)}(0, z)$ and $\frac{\delta F[p]}{\delta p(z)}(+\infty, z)$ as $\lim_{p(z)\downarrow 0}\frac{\delta F[p]}{\delta p(z)}(p(z), z)$ and $\lim_{p(z)\uparrow+\infty}\frac{\delta F[p]}{\delta p(z)}(p(z), z)$.
\begin{remark}
Although we need to define the functional derivative by the Fr{\'e}chet derivative or the G{\^a}teaux derivative \cite{frigyik2008introduction} for a more rigorous mathematical discussion, we adopt the above definition for simplicity.
\end{remark}
 
\begin{definition}[Lagrangian]
Let $F[p]: \P\rightarrow \mathbb{R}$ be an objective functional and $\Phi_i[p] \, (i=1,2,\cdots, m): \P\rightarrow \mathbb{R}$ be affine and equality constraint functionals.

Let $\Psi_j[p] \, (j=1,2,\cdots, n): \P\rightarrow \mathbb{R}$ be inequality constraint functionals.
The Lagrangian $L: \P\times\mathbb{R}^{m+1}\times\mathbb{R}^n \rightarrow\mathbb{R}$ is defined as 
\begin{align}
L[p](\lambda,\nu)\eqdef F[p] + \sum_{i=1}^{m+1} \lambda_i \Phi_i[p]+\sum_{j=1}^n \nu_j \Psi_j[p],
\end{align}
where $\lambda =(\lambda_1,\lambda_2, \cdots, \lambda_{m+1})^T\in\mathbb{R}^{m+1}$ and $\nu =(\nu_1,\nu_2, \cdots, \nu_n)^T\in\mathbb{R}^n$ are the Lagrange multipliers and $\Phi_{m+1}[p]\eqdef\int p(z)\dmu(z)-1$, which corresponds to the constraint $\int p(z)\dmu(z)=1$.
\end{definition}

\section{Main results and their proofs}                                                                        
\subsection{Convex optimization of strictly convex functionals} 
Let $\Phi_i[p] \, (i=1,2, \cdots m)\eqdef \int \varphi_i(z)p(z)\dmu(z)-c_i$, where $\varphi_i:\Omega\rightarrow \mathbb{R}$ and $(c_1,c_2,\cdots, c_m)^T\in \mathbb{R}^m$.                                

Let $\Psi_j[p] \, (j=1,2,\cdots, n)$ be differentiable convex functionals.  

We define the feasible set as $\D\eqdef\{p\in\P|\Phi_i[p]= 0 \, (i=1,2,\cdots, m)\mbox{\rm{ and }}\Psi_j[p]\leq 0 \, (j=1,2,\cdots, n)\}$. We can easily confirm that $\D$ is convex.

Consider the optimization problem of differentiable strictly convex functional $F[p]$.
\begin{align}
\mbox{ \rm{ minimize }} F[p] \mbox{\rm{ subject to }}p\in\D.
\end{align}

\begin{theorem}
\label{th_optimization}
Let $\Lambda\eqdef \{z\in\Omega|\exists \hat{p}(z)\in[0,\infty), \frac{\delta L[p]}{\delta p(z)}(\hat{p}(z),z,\lambda,\nu)=0\}$, where $L[p](\lambda,\nu)$ is the Lagrangian.
Suppose that there exist $p^\ast\in \D$, $\lambda\in\mathbb{R}^{m+1}$, and $\nu\in\mathbb{R}^n$ such that: 
\begin{align}
\label{th_optimization_eq1}
\mbox{\rm{1. }} &p^\ast(z) = \hat{p}(z), \mbox{\rm{ if }} z\in\Lambda, \\ 
\label{th_optimization_eq2}
&p^\ast(z) = 0, \mbox{\rm{ and }} \frac{\delta L[p]}{\delta p(z)}(0,z,\lambda,\nu)>0, \mbox{\rm{ if }} z\in\Omega\setminus\Lambda, \\ 
\label{th_optimization_eq3}
\mbox{\rm{2. }}&\nu_j\Psi_j[p^\ast]=0, \mbox{\rm{ and }} \nu_j\geq 0, \: j=1,2,\cdots, n.
\end{align}
Then, 
\begin{align}
\label{th_optimization_eq4}
\argmin_{p\in \D} F[p]=p^\ast,
\end{align}
and $p^\ast$ is unique.
\end{theorem}
The condition 1., 2., and $p^\ast\in \D$ correspond to the Karush-Kuhn-Tucker (KKT) conditions (see, e.g., Chapter 5 in \cite{boyd2004convex}) as will be described in the next subsection.
\begin{remark}
When $F[p]$ is convex (including affine), the condition 1. and 2. give the minimizer conditions. However, the minimizer need not be unique.
\end{remark}
\begin{corollary}
\label{cor_optimization}
Let $\Lambda\eqdef \{z\in\Omega|\exists \hat{p}(z)\in[0,\infty), \frac{\delta L[p]}{\delta p(z)}(\hat{p}(z),z,\lambda,\nu)=0\}$, where $L[p](\lambda,\nu)$ is the Lagrangian.

Let $\frac{\delta L[p]}{\delta p(z)}(p(z),z,\lambda,\nu)$ is continuous with respect to $p(z)$ for all $z\in\Omega$.
Suppose that there exist $p^\ast\in \D$, $\lambda\in\mathbb{R}^{m+1}$, and $\nu\in\mathbb{R}^n$ such that: 
\begin{align}
\label{cor_optimization_eq1}
\mbox{\rm{1. }} &p^\ast(z) = \hat{p}(z), \mbox{\rm{ if }} z\in\Lambda, \\ 
\label{cor_optimization_eq2}
&p^\ast(z) = 0, \mbox{\rm{ and }} \frac{\delta L[p]}{\delta p(z)}(+\infty,z,\lambda,\nu)>0, \mbox{\rm{ if }} z\in\Omega\setminus\Lambda, \\
\label{cor_optimization_eq3}
\mbox{\rm{2. }}&\nu_j\Psi_j[p^\ast]=0, \mbox{\rm{ and }} \nu_j\geq 0, \: j=1,2,\cdots, n.
\end{align}
Then, 
\begin{align}
\label{cor_optimization_eq4}
\argmin_{p\in \D} F[p]=p^\ast,
\end{align}
and $p^\ast$ is unique.
\end{corollary}
\begin{example}[Maximization of R{\'e}nyi entropy]
We show an application example of Corollary \ref{cor_optimization}.
For $\alpha \neq 1$, the R{\'e}nyi entropy \cite{renyi1961measures} is defined as 
\begin{align}
H_\alpha(p)\eqdef \frac{1}{1-\alpha}\log\int p^\alpha\dmu.
\end{align}

For simplicity, we discuss the case $\alpha > 1$ and $\Omega =\mathbb{R}$.
Johnson and Vignat \cite{johnson2007some} showed that the following function gives the maximum of the R{\'e}nyi entropy with linear equality constraints $\int zp(z)\dmu(z)=0$ and $\int z^2p(z)\dmu(z)=\sigma^2$ (see Proposition 1.3 in \cite{johnson2007some}).
\begin{align}
g_\alpha(z)=A_\alpha\biggl(1-(\alpha-1)\beta \frac{z^2}{\sigma^2}\biggr)_+^{\frac{1}{\alpha-1}}
\end{align}
with 
\begin{align}
\beta=\frac{1}{3\alpha-1},
\end{align}
and 
\begin{align}
A_\alpha=\Gamma\biggl(\frac{\alpha}{\alpha-1}+\frac{1}{2}\biggr)(\beta(\alpha-1))^{\frac{1}{2}}/\biggl(\Gamma\biggl(\frac{\alpha}{\alpha-1}\biggr)\pi^\frac{1}{2}\sigma\biggr).
\end{align}
Here $x_+=\max(x,0)$ denotes the positive part.
Maximizing the R{\'e}nyi entropy is equivalent to minimizing $F[p]=\int p^\alpha \dmu$, which is strictly convex for $\alpha>1$.
We show that $g_\alpha$ is also a minimizer under inequality constraints such that $\Psi_1[p]=\int zp(z)\dmu(z)\leq 0$ and $\Psi_2[p]=\int z^2p(z)\dmu(z)-\sigma^2\leq 0$
The Lagrangian is 
\begin{align}
L[p](\lambda,\nu)=F[p]+\lambda_1\biggl(1-\int p(z)\dmu(z)\biggr)+\nu_1\int zp(z)\dmu(z)+\nu_2\biggl(\int z^2p(z)\dmu(z)-\sigma^2\biggr).
\end{align}
The solution of $\frac{\delta L[p]}{\delta p(z)}(p(z),z,\lambda,\nu)=\alpha p(z)^{\alpha-1}-\lambda_1+\nu_1 z+\nu_2 z^2=0$ is 
\begin{align}
\hat{p}(z)=\biggl(\frac{1}{\alpha}(\lambda_1-\nu_1 z-\nu_2 z^2)\biggr)^\frac{1}{\alpha-1}.
\end{align}
From (\ref{cor_optimization_eq1}) and (\ref{cor_optimization_eq2}), we obtain
\begin{align}
p^\ast(z)=\biggl(\frac{1}{\alpha}(\lambda_1-\nu_1 z-\nu_2 z^2)\biggr)_+^\frac{1}{\alpha-1},
\end{align}
and $\Lambda =\{z\in\mathbb{R}|\lambda_1-\nu_1 z-\nu_2 z^2\geq 0\}$.
Since $\lim_{p(z)\uparrow+\infty} p(z)^{\alpha-1}=+\infty$, $\frac{\delta L[p]}{\delta p(z)}(+\infty,z,\lambda,\nu)>0$ in (\ref{cor_optimization_eq2}) is satisfied.
By choosing $\lambda_1=\alpha A_\alpha^{\alpha-1}>0$, $\nu_1=0$ and $\nu_2=(\alpha-1)\beta\lambda_1/\sigma^2\geq 0$, $p^\ast$ satisfies $\Psi_1[p^\ast]=\Psi_2[p^\ast]=0$.
Hence, $p^\ast$, $\nu_1$, and $\nu_2$ satisfy (\ref{cor_optimization_eq3}) and $p^\ast=g_\alpha$ is the unique minimizer. 
\end{example}

\subsection{Proofs of main results}
\begin{lemma}
\label{lem_convex}
Let $G:\P\rightarrow \mathbb{R}$ be a differentiable convex functional and $p,q\in\P$.
Then,
\begin{align}
G[q]\geq G[p]+\int  \frac{\delta G[p]}{\delta p(z)}(p(z),z)(q(z)-p(z)) \dmu(z).
\end{align}
\end{lemma}

\begin{proof}
From (\ref{def_eq_strictly_functional1}), it follows that
\begin{align}
\label{eq_lem_convex1}
G[q]-G[p]\geq \frac{G[(q-p)t+p]-G[p]}{t},
\end{align}
where $t\in(0,1]$.
In the limit $t\downarrow 0$, from the differentiability of $G$ and the definition of the functional derivative for $\eta=q-p$, we have
\begin{align}
\label{eq_lem_convex2}
\lim_{t\downarrow 0}\frac{G[(q-p)t+p]-G[p]}{t}=\left.\frac{d}{d\epsilon}G[p+\epsilon (q-p)]\right|_{\epsilon=0}=\int  \frac{\delta G[p]}{\delta p(z)}(p(z),z)(q(z)-p(z)) \dmu(z).
\end{align}
By combining (\ref{eq_lem_convex1}) and (\ref{eq_lem_convex2}), the result follows.
\end{proof}

\begin{proof}[Proof of Theorem \ref{th_optimization}]                                                                              
First, we prove the uniqueness of the minimizer.
Suppose that there exist two different minimizers $p^\ast_1$ and $p^\ast_2$.
Since $\D$ is a convex set, it follows that $\frac{1}{2}p^\ast_1+\frac{1}{2}p^\ast_2\in\D$.
From the strictly convexity of $F$, it follows that 
\begin{align}
F[p^\ast_1]=\frac{1}{2}F[p^\ast_1]+\frac{1}{2}F[p^\ast_2]>F\biggl[\frac{1}{2}p^\ast_1+\frac{1}{2}p^\ast_2\biggr],
\end{align}
where we use $F[p^\ast_1]=F[p^\ast_2]$.
This contradicts that $p^\ast_1$ is a minimizer.
Hence, the minimizer is unique.

Next, we prove that $p^\ast$ is a minimizer.
We introduce a function $\theta: \Omega\rightarrow\mathbb{R}$ as follows.
\begin{numcases}
{}
\label{proof_th1_eq1}
\theta(z) = 0, \mbox{\rm{ if }} z\in\Lambda, &\\ 
\label{proof_th1_eq2}
\theta(z) = \frac{\delta L[p]}{\delta p(z)}(0,z,\lambda,\nu) > 0, \mbox{\rm{ if }} z\in\Omega\setminus\Lambda. &
\end{numcases}
The inequality in the right hand side in (\ref{proof_th1_eq2}) comes from (\ref{th_optimization_eq2}).
We define a modified Lagrangian as
\begin{align}
\tilde{L}[p](\lambda,\nu)\eqdef L[p](\lambda,\nu)-\int \theta(z)p(z)\dmu(z).
\end{align}
$\theta(z)$ are the Lagrange multipliers for the constraints $p(z)\geq 0$ for all $z\in\Omega$.
From the convexity of $F$ and $\Psi_j$, the linearity of $\Phi_i$, and $\nu_j\geq 0$, the modified Lagrangian $\tilde{L}[p](\lambda,\nu)$ is a convex functional.
By applying Lemma \ref{lem_convex}, it follows that
\begin{align}
\label{proof_th1_eq3}
\tilde{L}[q](\lambda,\nu)\geq \tilde{L}[p^\ast](\lambda,\nu)+\int  \frac{\delta \tilde{L}[p]}{\delta p(z)}(p^\ast(z),z,\lambda,\nu)(q(z)-p^\ast(z)) \dmu(z),
\end{align}
where $q$ is an arbitrary probability density in $\D$.
From the definition of $\tilde{L}$, we have
\begin{align}
\label{proof_th1_eq4}
\frac{\delta\tilde{L}[p]}{\delta p(z)}(p(z),z,\lambda,\nu)=\frac{\delta L[p]}{\delta p(z)}(p(z),z,\lambda,\nu)-\theta(z).
\end{align}
When $z\in\Lambda$, from (\ref{th_optimization_eq1}) and (\ref{proof_th1_eq1}), it follows that 
\begin{align}
\frac{\delta\tilde{L}[p]}{\delta p(z)}(p^\ast(z),z,\lambda,\nu)=\frac{\delta L[p]}{\delta p(z)}(p^\ast(z),z,\lambda,\nu)=0.
\end{align}
When $z\in\Omega\setminus\Lambda$, from (\ref{th_optimization_eq2}) and (\ref{proof_th1_eq2}) it follows that 
\begin{align}
\frac{\delta\tilde{L}[p]}{\delta p(z)}(p^\ast(z),z,\lambda,\nu)=\frac{\delta L[p]}{\delta p(z)}(0,z,\lambda,\nu)-\theta(z)=0.
\end{align}
From these equalities, it follows that
\begin{align}
\label{proof_th1_eq5}
\frac{\delta\tilde{L}[p]}{\delta p(z)}(p^\ast(z),z,\lambda,\nu)=0
\end{align}
for all $z\in\Omega$.
Substituting (\ref{proof_th1_eq5}) into (\ref{proof_th1_eq3}) gives that 
\begin{align}
\label{proof_th1_eq6}
\tilde{L}[q](\lambda,\nu)\geq \tilde{L}[p^\ast](\lambda,\nu).
\end{align}
By combining (\ref{th_optimization_eq2}) and (\ref{proof_th1_eq1}), it follows that
\begin{align}
\label{proof_th1_eq7}
\theta(z)p^\ast(z)=0
\end{align}
for all $z\in\Omega$.
By combining this equality, (\ref{th_optimization_eq3}), and the definitions of $\D$, it follows that
\begin{align}
\label{proof_th1_eq8}
\tilde{L}[p^\ast](\lambda,\nu)=F[p^\ast] + \sum_{i=1}^{m+1} \lambda_i \Phi_i[p^\ast] +\sum_{j=1}^n \nu_j \Psi_j[p^\ast] -\int \theta(z)p^\ast(z)\dmu(z)=F[p^\ast],
\end{align}
where we use $\Phi_i[p^\ast]=\nu_j \Psi_j[p^\ast]=0$ for all $i$ and $j$.

On the other hand, by combining (\ref{proof_th1_eq1}), (\ref{proof_th1_eq2}), and $q(z)\geq 0$, it follows that
\begin{align}
\label{proof_th1_eq9}
\theta(z)q(z)\geq 0.
\end{align}
By combining this equality, (\ref{th_optimization_eq3}), and the definitions of $\D$, it follows that
\begin{align}
\label{proof_th1_eq10}
\tilde{L}[q](\lambda,\nu)=F[q] +\sum_{i=1}^{m+1} \lambda_i \Phi_i[q] +\sum_{j=1}^n \nu_j \Psi_j[q] -\int \theta(z)q(z)\dmu(z)\leq F[q],
\end{align}
where we use  $\Phi_i[q]=0$ and $\nu_j \Psi_j[q]\leq 0$ for all $i$ and $j$.
By combining (\ref{proof_th1_eq6}),  (\ref{proof_th1_eq8}) and  (\ref{proof_th1_eq10}), it follow that 
\begin{align}
F[q]\geq F[p^\ast].
\end{align}
Since $q$ is an arbitrary probability $\D$, $p^\ast$ is a minimizer.
\end{proof}

From (\ref{th_optimization_eq3}), (\ref{proof_th1_eq1}), (\ref{proof_th1_eq2}),  
 (\ref{proof_th1_eq5}), (\ref{proof_th1_eq7}), and $p^\ast\in\D$, $p^\ast$, $\lambda_i$, $\nu_j$, and $\theta(z)$ satisfy the following KKT conditions.
\begin{align}
\mbox{\rm{a. }} &p^\ast\in\D, \\ 
\mbox{\rm{b. }} &\frac{\delta\tilde{L}[p]}{\delta p(z)}(p^\ast(z),z,\lambda,\nu)=0, \\
\mbox{\rm{c. }} &\nu_j\Psi_j[p^\ast]=0, \mbox{\rm{ and }} \nu_j\geq 0, \: j=1,2,\cdots, n \\
\mbox{\rm{d. }} &\theta(z)p^\ast(z)=0, \mbox{\rm{ and }} \theta(z) \geq 0, \,\forall z \in\Omega.
\end{align}

\begin{proof}[Proof of Corollary \ref{cor_optimization}]
Suppose that $\frac{\delta L[p]}{\delta p(z)}(0,z_0,\lambda,\nu)\leq0$ for $z_0\in\Omega\setminus\Lambda$.
From the assumptions, since $\frac{\delta L[p]}{\delta p(z)}(+\infty,z_0,\lambda,\nu)>0$ and $\frac{\delta L[p]}{\delta p(z)}(p(z_0),z_0,\lambda,\nu)$ is continuous with respect to $p(z_0)$, there exist $\hat{p}(z_0)\in[0,\infty)$ that satisfies $\frac{\delta L[p]}{\delta p(z)}(\hat{p}(z_0),z_0,\lambda,\nu)=0$.
This contradicts the definition of $\Lambda$.
Hence, the minimizer conditions in Theorem \ref{th_optimization} and Corollary \ref{cor_optimization} are equivalent.
\end{proof}
\section{Examples of strictly convex functionals}
We show some examples of differentiable strictly convex functionals.
These functionals can be inequality constraints as well as objective functionals.
Let $P,Q$ are probability measures and $p,q$ are $\mu$-densities of $P,Q$.
\begin{example}[Shannon entropy]
The Shannon entropy is defined as 
\begin{align}
h(p)\eqdef -\int p\log p\dmu. \nonumber
\end{align}
The negative Shannon entropy is strictly convex functional.
Let $F[p]=-h(p)$.
The functional derivative is 
\begin{align}
\frac{\delta F[p]}{\delta p(z)}(p(z))=\log p(z) +1.\nonumber
\end{align} 
\end{example} 

\begin{example}[Relative entropy]
The relative entropy (Kullback-Leibler divergence) is defined as
\begin{align}
D(P\|Q)\eqdef \int p\log\frac{p}{q} \dmu.\nonumber
\end{align}
The relative entropy is strictly convex functional in both arguments.
Let $F[p]=D(P\|Q)$ and $F[q]=D(P\|Q)$.
The functional derivative is 
\begin{align}
\frac{\delta F[p]}{\delta p(z)}(p(z))=\log \frac{p(z)}{q(z)} + 1, \nonumber \\ 
\frac{\delta F[q]}{\delta q(z)}(q(z))=-\frac{p(z)}{q(z)}. \nonumber
\end{align} 
\end{example}

\begin{example}[$f$-divergence]
Let $f: \mathbb{R}\rightarrow \mathbb{R}$ be a strictly convex function and $f(1)=0$.
The $f$-divergence is defined as
\begin{align}
D_f(P\|Q)\eqdef \int qf\biggl(\frac{p}{q}\biggr) \dmu. \nonumber
\end{align}
The $f$-divergence is strictly convex in both arguments and include the relative entropy.
Let $F[p]=D_f(P\|Q)$ and $F[q]=D_f(P\|Q)$.
If $f$ is differentiable, the functional derivatives are
\begin{align}
\frac{\delta F[p]}{\delta p(z)}(p(z))=f'\biggl(\frac{p(z)}{q(z)}\biggr), \nonumber
\end{align} 
\begin{align}
\frac{\delta F[q]}{\delta q(z)}(q(z))=\tilde{f}'\biggl(\frac{q(z)}{p(z)}\biggr), \nonumber
\end{align} 
where $\tilde{f}(x)\eqdef xf\bigl(\frac{1}{x}\bigr)$.
\end{example}

\begin{example}[Bregman divergence \cite{bregman1967relaxation}]
\label{ex_bregman}
Let $f: \mathbb{R}\rightarrow \mathbb{R}$ be a differentiable and strictly convex function.
The Bregman divergence is defined as
\begin{align}
D_{\mathrm{B}}(P\|Q)\eqdef \int f(p) \dmu-\int f(q)\dmu -\int f'(q)(p-q) \dmu, \nonumber
\end{align}
where $f'(x)$ denotes the derivative of $f$.
The Bregman divergence is strictly convex in the first argument.
Let $F[p]=D_{\mathrm{B}}(P\|Q)$.
The functional derivative is 
\begin{align}
\frac{\delta F[p]}{\delta p(z)}(p(z))=f'(p(z))-f'(q(z)). \nonumber
\end{align} 
\end{example}

\begin{example}[R{\'e}nyi-divergence \cite{van2014renyi}]
For $0<\alpha<\infty$, the R{\'e}nyi-divergence is defined as
\begin{align}
D_\alpha(P\|Q)\eqdef \frac{1}{\alpha-1}\log\int p^\alpha q^{1-\alpha} \dmu \mbox{ for } \alpha \neq 1, \nonumber \\
D_1(P\|Q)\eqdef\int p\log\frac{p}{q} \dmu. \nonumber
\end{align}
The R{\'e}nyi divergence is strictly convex in the second argument for $0<\alpha<\infty$ (see \cite{van2014renyi}).
Let $F[q]=D_\alpha(P\|Q)$.
The functional derivative is
\begin{align}
\frac{\delta F[q]}{\delta q(z)}(q(z))=-\frac{1}{\int p^\alpha q^{1-\alpha} \dmu}{\biggl(\frac{p(z)}{q(z)}\biggr)}^\alpha. \nonumber
\end{align} 
\end{example}

\section{Conclusion}
We have discussed the convex optimization problems on strictly convex functionals of probability densities. We have shown the conditions of minimizers and the uniqueness of minimizer if there exist a minimizer.                                                                                                            
The conditions of minimizers are \\
1. The minimizer $p^\ast(z)$ is equal to the stationary point of the Lagrangian if the stationary point is non-negative. \\
2. If the stationary point is not a non-negative real number, $p^\ast(z)=0$ and the functional derivative of the Lagrangian at $p^\ast(z)=0$ is positive. \\
3. The Lagrange multipliers corresponding to inequality constraints is non-negative, and the products of the Lagrange multipliers and the functionals corresponding to the inequality constraint are equal to $0$.

\bibliography{reference_optimization}
\end{document}